\documentclass[]{article}

\newenvironment{frontmatter}{}{\maketitle}
\newenvironment{keyword}{\emph{Keywords: }}{}
\newcommand{\sep}{;\ }

\usepackage[numbers]{natbib}
\usepackage{enumerate}

\usepackage{ucs}
\usepackage[utf8x]{inputenc}
\usepackage[english]{babel}
\usepackage{graphicx}
\usepackage{url}
\usepackage{subfigure}
\usepackage{amsthm}
\usepackage{amsmath}
\usepackage{amssymb}
\usepackage{slashbox}
\usepackage{hyperref}
\usepackage{xcolor}
\usepackage{xspace}
\usepackage{mathrsfs}

\PrerenderUnicode{ű}
\PrerenderUnicode{É}

\newtheorem{theo}{Theorem}
\newtheorem{prop}{Proposition}
\newtheorem*{prop*}{Proposition}
\newtheorem{coro}{Corollary}
\theoremstyle{definition}
\newtheorem{defn}{Definition}

\newcommand{\suh}{$\mathcal{I}^+$\xspace}
\newcommand{\suhm}{\mathcal{I}^+}
\newcommand{\col}{$\mathcal{G}^+$\xspace}
\newcommand{\colm}{\mathcal{G}^+}
\newcommand{\Gomboc}{Gömböc\xspace}
\newcommand{\plantri}{{\tt plantri}\xspace}

\newcommand{\sdeg}{m} 

\newcommand{\maps}{|\mathcal{I}^+|}
\newcommand{\quads}{|\mathscr{Q}|}
\newcommand{\sdquads}{|\mathscr{Q}_\text{SD}|}
\renewcommand{\maps}{e}
\renewcommand{\quads}{q}
\renewcommand{\sdquads}{e_\text{SD}}

\newcommand{\kikomment}[1]{} 

\newcommand{\nfour}{}
\newcommand{\neight}{}
\newcommand{\ntwelve}{}
\newcommand{\nsteen}{}

\begin{document}


\begin{frontmatter}

\title{Generating spherical multiquadrangulations by restricted vertex splittings and the reducibility of equilibrium classes}

\author{
Richárd Kápolnai\footnote{\url{kapolnai@iit.bme.hu},
Dept.\ of Control Engineering and Information Technology, Budapest University of Technology and Economics, H-1117 Magyar tudósok körútja 2., Budapest, Hungary
}
\and
Gábor Domokos\footnote{\url{domokos@iit.bme.hu},
Dept.\ of Mechanics, Materials \& Structures, Budapest University of Technology and Economics, H-1521, Műegyetem rakpart 1-3.\ K.II.42., Budapest, Hungary
}
\and
Tímea Szabó\footnote{\url{tszabo@szt.bme.hu},
Dept.\ of Mechanics, Materials \& Structures, Budapest University of Technology and Economics, H-1521, Műegyetem rakpart 1-3.\ K.II.42., Budapest, Hungary
}
}




\end{frontmatter}


\begin{abstract}

A spherical quadrangulation is a loopless graph embedded on the sphere such that each face is bounded by a walk of length 4, parallel edges allowed.
The family of the isomorphism classes of quadrangulations can be generated by a sequence of graph operations called vertex splitting, starting from the path graph $P_2$ with three vertices and two edges.
$P_2$ is also referred to as the common ancestor of all quadrangulations.
We define the degree $1\leq D \leq \lfloor d/2\rfloor$ of a splitting $S$ (where $d$ is the degree of the split vertex) and consider \em restricted \rm splittings $S_{i,j}$ with $1\leq i\leq D \leq j\leq \lfloor d/2\rfloor$.
As Brinkmann and coworkers have recently pointed out,
 restricted splittings $S_{2,3}$ generate all simple quadrangulations.

Here we investigate the cases $S_{1,2},S_{1,3},S_{1,1},S_{2,2},S_{3,3}$.
First we show that the restricted splittings $S_{1,2}$ are exactly the \emph{monotone} ones in the sense that the resulting graph contains the original as a subgraph.
Then we proceed to show that they define a set of nontrivial ancestors beyond $P_2$ and each quadrangulation has a unique ancestor.

Our results have a direct geometric interpretation in the context of mechanical equilibria of convex bodies.
The latter can be defined as the scalar distance $R(\theta,\varphi)$ measured from the
center of gravity and the Morse-Smale complex associated with the gradient of $R$ corresponds to a
2-coloured quadrangulation with independent set sizes $s,u$.
The numbers $s,u$ of coloured vertices identify the \em primary equilibrium class \rm associated with the body by Várkonyi and Domokos.
We show that the $S_{1,1}$
and $S_{2,2}$ splittings generate all primary equilibrium classes
(in case of $S_{1,1}$ from a single ancestor, in case of $S_{2,2}$
from a finite nontrivial set of ancestors). This is closely related to the geometric results of Várkonyi and Domokos where they show that specific geometric transformations can generate
all equilibrium classes.

If, beyond the  numbers $s,u$, the full topology of the
quadrangulation is considered, we arrive at the more refined
\em secondary equilibrium classes\rm.
As Domokos, Lángi and
Szabó showed recently,
 one can create the geometric
counterparts of unrestricted splittings to generate all secondary classes.
Our results show that restricted, monotone
splittings $S_{1,2}$, while adequate to generate all primary classes
from one single ancestor,
can only generate a limited range of secondary equilibrium classes from the same ancestor.
The geometric interpretation of the additional ancestors defined
by monotone splittings shows that minimal polyhedra play a key
role in this process.
We also present some computational results on the cardinality of secondary equilibrium classes and multiquadrangulations.

\end{abstract}











\section{Introduction}
\label{sec:intro}

Our work is motivated by the classification system of convex, homogeneous 3D bodies introduced by \citet{Domokos2006}.
They map each body its \emph{primary equilibrium class} defined by the \emph{numbers} of the stable and unstable equilibrium points of the body surface.
Moreover, the isomorphism classes of the \emph{topologies} of the equilibria provides a refined, secondary classification system, where such a topology can be genuinely represented by a 2-coloured quadrangulation \cite{Domokos2012}.

\subsection{Generating multiquadrangulations}

A \emph{quadrangulation} of the sphere is a loopless graph embedded in the sphere having every face bounded by a closed walk of length 4.
We allow parallel edges, and the boundary walk may repeat edges or vertices.
This definition was also used by~\citet{Mohar2010}, however, \citet{Archdeacon2001} applied the word ``pseudoquadrangulation'' instead for multigraphs.
If we want to emphasize that the quadrangulation may have parallel edges, it is called a \emph{multiquadrangulation},
if a quadrangulation has no parallel edges, it is called a \emph{simple quadrangulation}.
Note that the 2-path $P_2$ (the path of length 2 with two edges and three vertices) is the smallest quadrangulation, and the 4-cycle $C_4$ (the cycle of length 4) is the smallest simple quadrangulation, illustrated on \autoref{fig:smallests}.

Two quadrangulations are considered isomorphic, if there is a homeomorphism from one to another that either preserves or reverses the orientation of the embedding.
Equivalently, the \emph{cyclic ordering} of the incident edges at each vertex is either preserved or reversed, so e.g.\  a graph is isomorphic to its reflection.
Such an isomorphism class is also called an \emph{unsensed, unrooted map} in the literature \cite{Walsh2007}.
Let $\mathscr{Q}$ denote the family of all multiquadrangulations, and $\mathscr{Q}_1$ the family of all simple quadrangulations.

We say a graph family $\mathscr{F}$ is \emph{generated} from the starting set ${K}\subset\mathscr{F}$ by some given graph operations, if each graph in $\mathscr{F}$ can be constructed from some graph of ${K}$ by applying a finite series of the given graph operations.
All our graph operations are based on the vertex splitting, depicted on \autoref{fig:splitting}, explained as follows.
\begin{figure}
  \centering
  \subfigure[$\sdeg>1$]{
    \label{fig:msplitting}
    \includegraphics{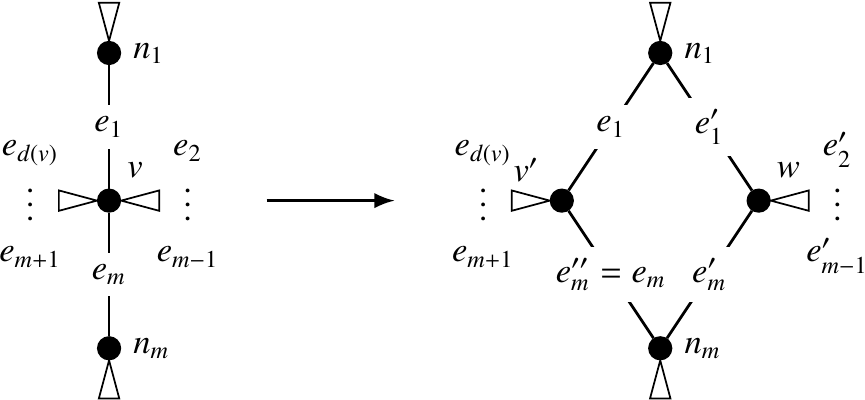}
  }
  \qquad
  \subfigure[$\sdeg=1$]{
    \label{fig:1splitting}
    \includegraphics{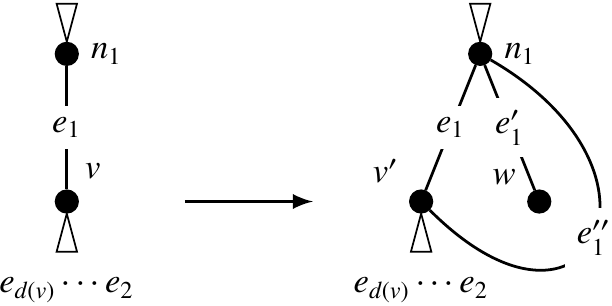}
  }
  \caption{Vertex splitting of degree $m$.}
  \label{fig:splitting}
\end{figure}
The embedding (the cyclic ordering) of the graphs showed is also important so the small open triangles denote that other edges \emph{may} occur only at that position so they clarify the cyclic ordering of the edges in the original (left to the arrow) and in the resulting (right to the arrow) graph.
Vertex splitting is also known in the literature as the inverse of face contraction.

\citet{Batagelj1989, Negami1993} showed that the splittings generate $\mathscr{Q}_1$ from $C_4$.
There are a number of related results regarding inductive generation of certain simple quadrangulation families, e.g.\ 
\citet{Batagelj1989} gave inductive definition of 3-connected quadrangulations,
\citet{Nakamoto1999} generated quadrangulations with minimum degree 3, 
\citet{Brinkmann2005} improved efficiency for all these families above.
While \cite{Batagelj1989, Negami1993, Nakamoto1996b} mainly focus on simple quadrangulations, the following observation is a straightforward extension of their results: the splitting generalized for parallel case as well ($\sdeg=1$ on \autoref{fig:1splitting}) generates $\mathscr{Q}$ from $P_2$.

As shown on \autoref{fig:splitting}, the splittings replace a vertex $v$ with vertices $w$ and $v'$ dividing the edges of $v$.
The \emph{degree of a splitting} 
 is $D:=\min\{d(v'),d(w)\}$ where $d(v)$ denotes the degree of the vertex $v$.
Note that $1\leq D\leq \lfloor d(v)/2 \rfloor$ always holds.
We consider \emph{restricted splittings} $S_{i,j}$ with a limited range of $D$ such that $1\leq i\leq D \leq j\leq \lfloor d/2\rfloor$.
\citet{Brinkmann2005} showed that the restricted splittings $S_{2,3}$ (with $2\leq D\leq 3$) are enough to generate $\mathscr{Q}_1$ from $C_4$.
Their result can be easily extended for parallel graphs using the generalized vertex splitting, so our previous observation can be improved as follows:

\begin{theo}
\label{theo:splitting}
The restricted splittings $S_{1,3}$ generate $\mathscr{Q}$ from $P_2$. 
\end{theo}

In this paper we investigate the hierarchy generated by restricting the splittings $S_{1,2}$ (see \autoref{fig:monotone}), and we also show some application of restrictions $S_{1,1}$, $S_{2,2}$ and $S_{3,3}$.
\begin{figure}
  \centering
    \includegraphics{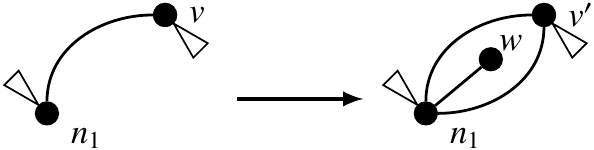}

    \includegraphics{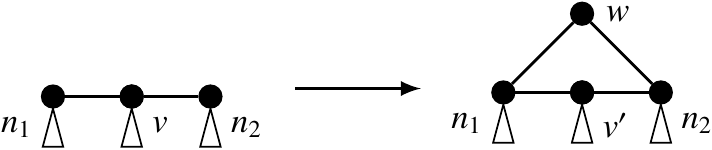}
    \qquad
    \includegraphics{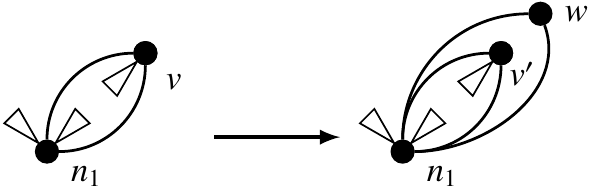}
  \caption{The monotone vertex splittings $S_{1,2}$. The 2-splitting has two variants: simple case and parallel case.}
  \label{fig:monotone}
\end{figure}
The main reason we focus on splittings $S_{1,2}$ is that these are not only local modifications but they purely extend the graph without removing any edge.
Formally, we say a graph operation is \emph{monotone} if the original graph is the embedded subgraph of the resulting one.
Actually $S_{1,2}$ are the only monotone operations on quadrangulations introducing only one new vertex.
The local monotone operations are also called \emph{face subdivisions} in the literature meaning one face is divided into smaller regions (e.g.\  \cite{Tutte1963,Mohar2010}).
\citet{Tutte1963} called the \emph{order} of a face subdivision the number of the introduced vertices in the operation, although he applied this concept on triangulations.
So on quadrangulations the monotone vertex splittings are exactly the face subdivisions of order 1.

For the splitting $S_{D,D}$ with a given $D$, we use the shorthand \emph{$D$-splitting}, and we call its inverse \emph{$D$-contraction}.
If a $D$-contraction is applicable to a graph, we say the graph is \emph{$D$-contractible}, otherwise we say the graph is \emph{$D$-irreducible}.
It is known that among the simple quadrangulations, $C_4$ is the only graph which is $D$-irreducible for any $D$ \cite{Batagelj1989, Negami1993}.
It follows from \autoref{theo:splitting}
 that among the multiquadrangulations, only $P_2$ is $D$-irreducible for any $D$, hence $P_2$ is also referred to as the only ancestor of all quadrangulations.
However, restricting the splittings to $S_{1,2}$ also admits other, nontrivial ancestors.
Throughout this paper we say a graph is an \emph{irreducible ancestor} (or shortly \emph{irreducible}) if it is 1-irreducible and 2-irreducible.
We mention that the concept of irreducible graphs with respect to some given graph operations was analogously used in e.g.\  \cite{Negami1993, Nakamoto1996b}.

We characterise the irreducible graphs as the ones with minimum degree 3, and extend a result of \citet{Batagelj1989} that the small graphs with less than 8 vertices are not irreducible, except $P_2$:

\begin{theo}
  \label{theo:gomboc}
  Every quadrangulation with less than 8 vertices is generated from $P_2$ by monotone splittings.
\end{theo}

We strengthen this theoretical result with observations on a data set generated by our computer program which we developed to explicitly enumerate every possible quadrangulation for a given size, based on the software \plantri{} \cite{Brinkmann2007}.
Furthermore, our data set also shows that there are only three irreducible graphs up to a size of 10 vertices.
Besides $P_2$, the other ones are the radial graphs (a.k.a.\  vertex-face incidence graphs \cite{Mohar2001}) of the skeletons of the two smallest polyhedra.
The \emph{radial graph} $R(G)$ of an embedded graph $G$ is a bipartite embedded graph such that one of the independent vertex sets of $R(G)$ corresponds to the vertex set of $G$, the other one to the face set of $G$.
Two vertices are connected in $R(G)$ with the same multiplicity as the incidence multiplicity of their preimages in $G$ (i.e.\  the appearance count of a vertex in the boundary walk of a face), with the obvious cyclic orderings.
In addition, we prove

\begin{theo}
\label{theo:minpol}
The radial graph of any polyhedral skeleton is irreducible, and they are generated from the radial graphs of the pyramids with polygonal base by the restricted splitting $S_{3,3}$.
\end{theo}

Then we prove that the irreducible ancestor of any graph is unique, i.e.\  a graph cannot be generated starting from two different irreducible graphs:

\begin{theo}
  \label{theo:ancestor}
  If a quadrangulation can be generated from both irreducible ancestors $A_1$ and $A_2$ by monotone splittings, then $A_1$ is isomorphic to $A_2$.
\end{theo}

\subsection{Generating secondary equilibrium classes}

The main goal of this paper is twofold: one is to contribute some observations on the reachability of quadrangulations i.e.\  which quadrangulations can be generated from others, the other is to apply them on the geometric interpretation mentioned earlier.
So finally, we translate our observations into the context of topology of the equilibrium points of convex bodies.
A generic convex body is given by its scalar height function $R(\theta, \varphi)$ which gives the distance between the surface and the mass center for any direction.
A surface point is an \emph{equilibrium} if the gradient of the height function $R$ is zero at that point (illustrated on the ellipsoid on \autoref{fig:horbit}).
\begin{figure}
  \subfigure[Height function $R$]{
    \label{fig:ellfel}
    \includegraphics[width=5cm]{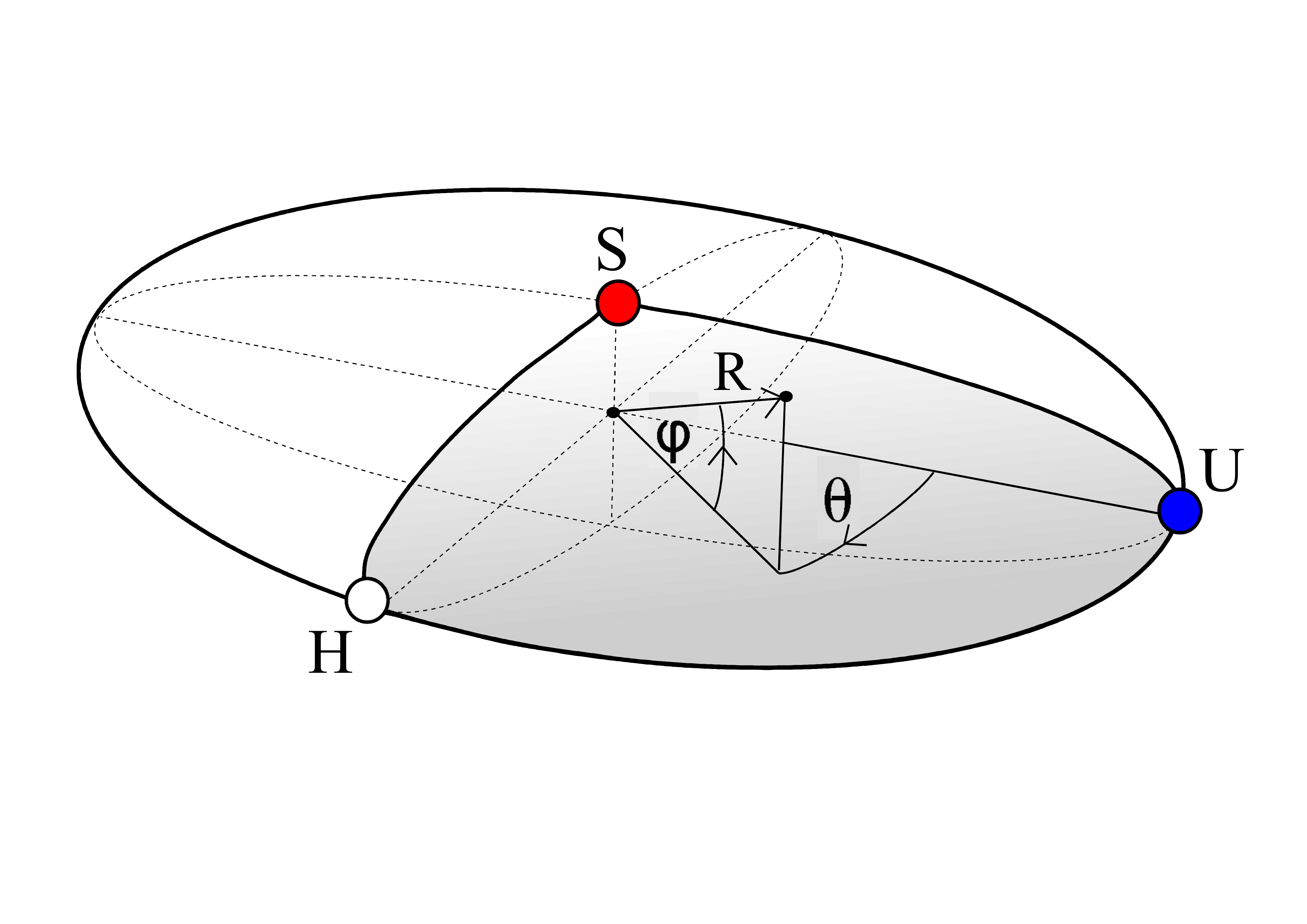}
  }
  \subfigure[Gradient field of $-R$]{
    \label{fig:ellgrads}
    \includegraphics[width=5cm]{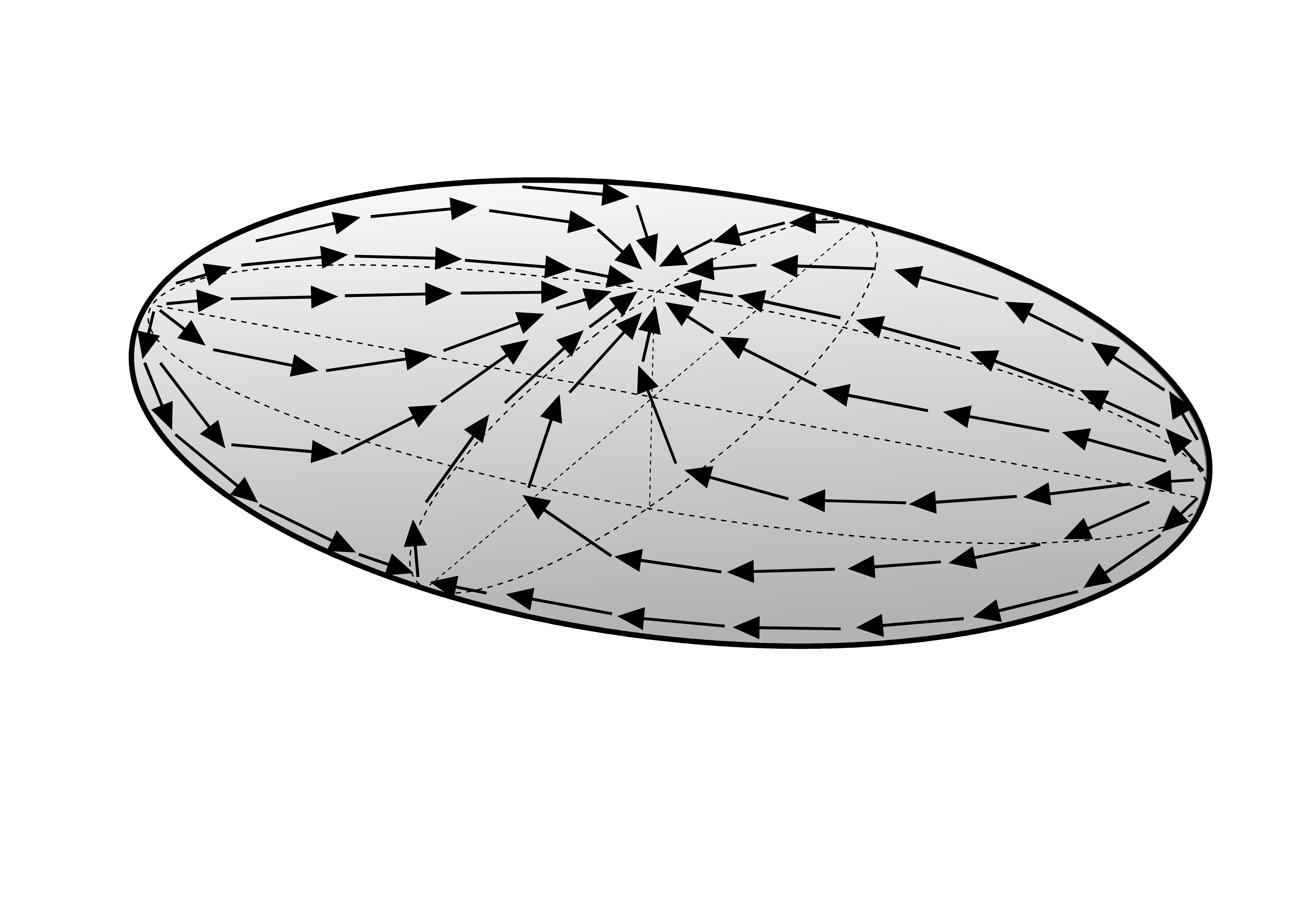}
  }
  \subfigure[Heteroclinic orbits]{
    \label{fig:ellorbits}
    \includegraphics[width=5cm]{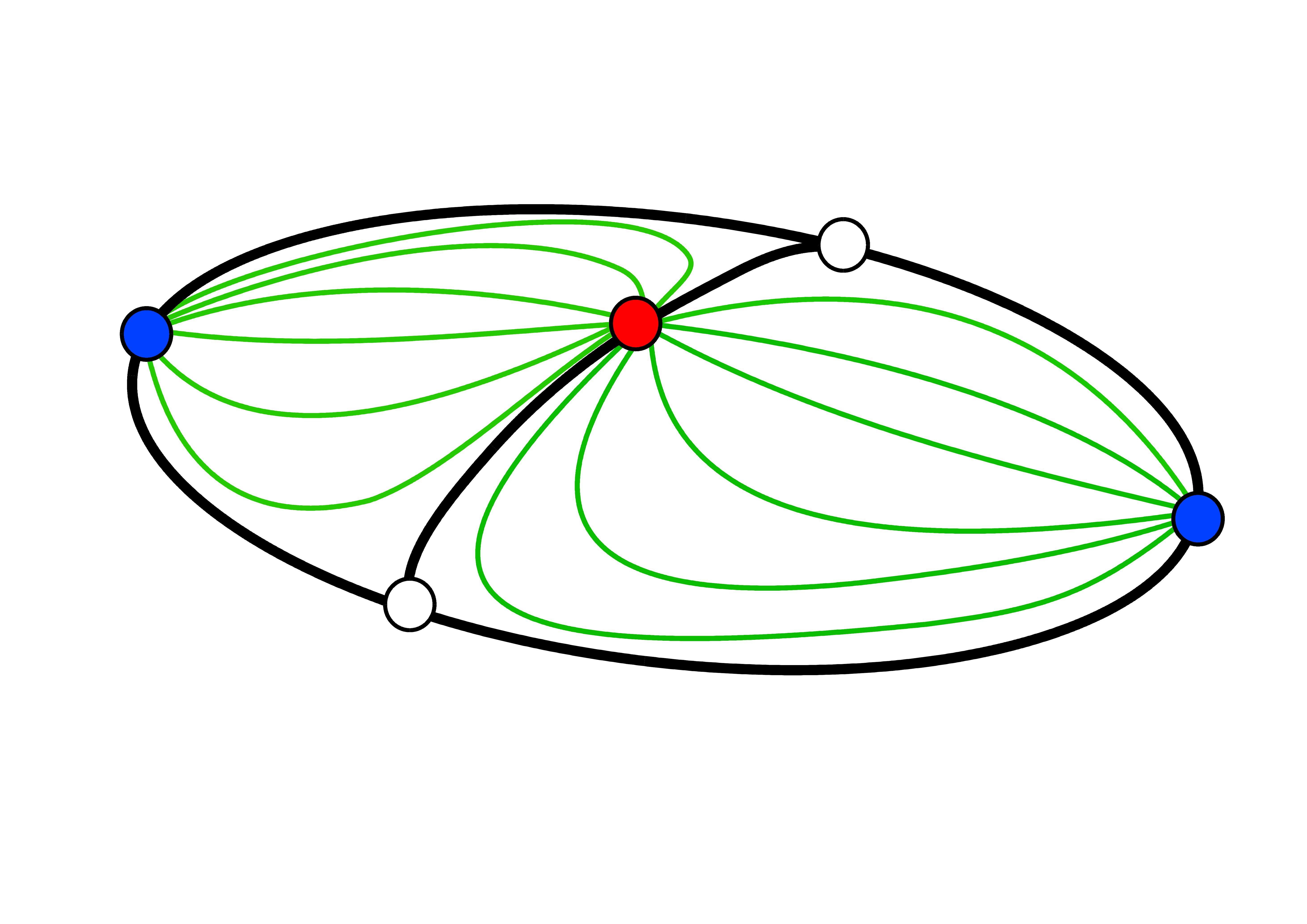}
  }

  \subfigure[Height function $R$]{
    \label{fig:ellmag}
    \includegraphics[width=5cm]{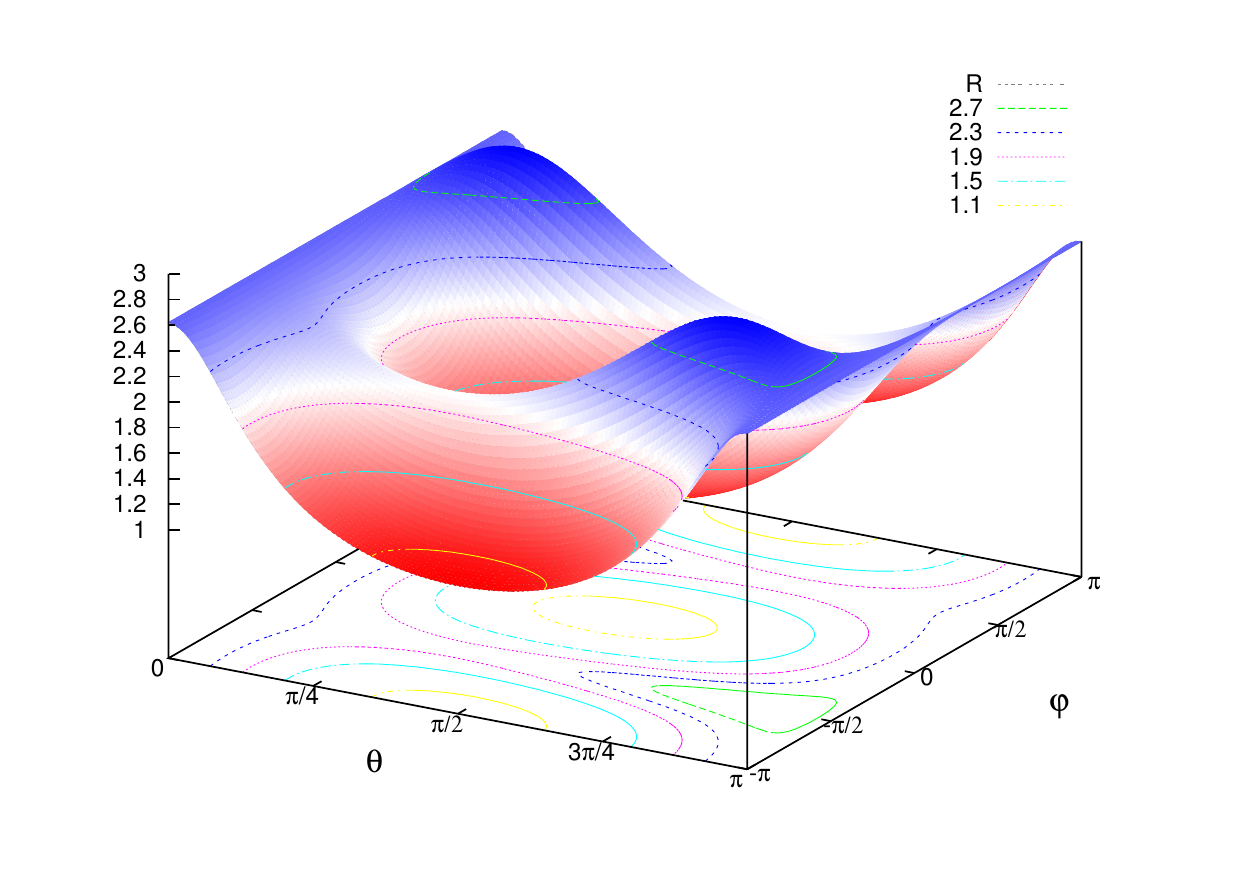}
  }
  \subfigure[Gradient field of $-R$]{
    \label{fig:ellgrad}
    \includegraphics[width=5cm]{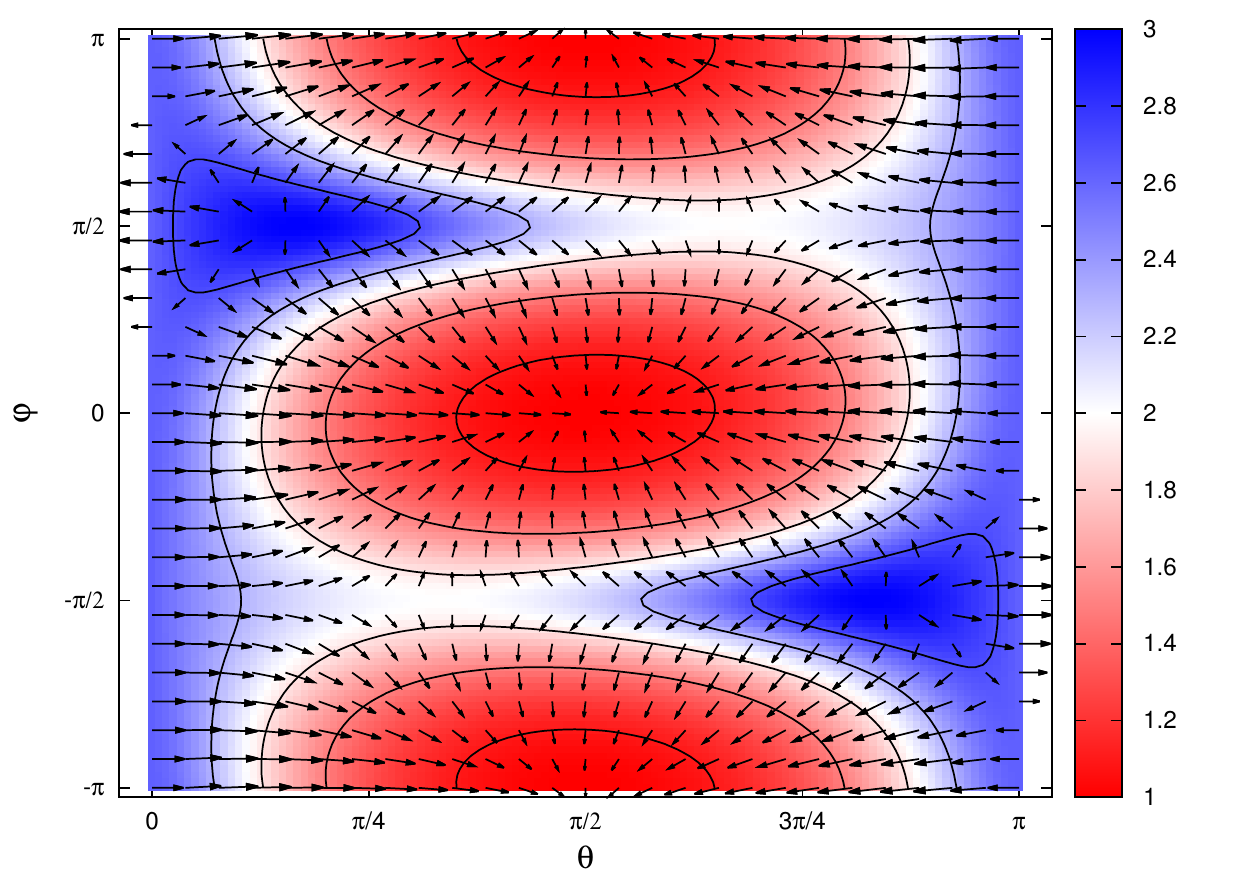}
  }
  \subfigure[Heteroclinic orbits]{
    \label{fig:ellorb}
    \includegraphics[width=5cm]{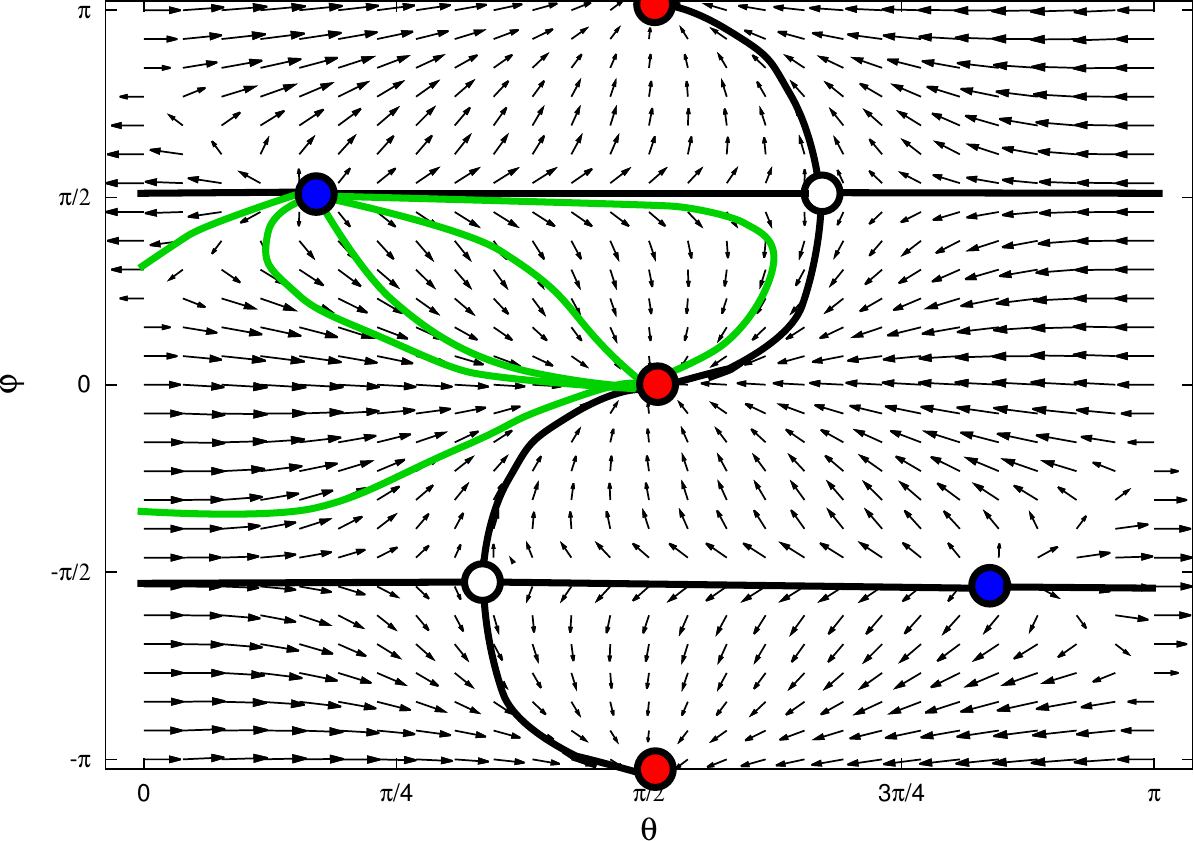}
  }

  \caption{Morse--Smale graph of the ellipsoid.
The axes of the ellipsoid are rotated for presentation purposes.
Colouring notation:  red: stable, blue: unstable, white: saddle.
}
  \label{fig:horbit}
\end{figure}
An equilibrium is \emph{stable} if $R$ takes a local minimum value, \emph{unstable} if $R$ takes a local maximum value, otherwise it is a \emph{saddle point}.
According to the Poincaré--Hopf theorem~\cite{Arnold1978}, the cardinality of stable, unstable and saddle equilibria of a body, denoted respectively by $s,u,h$, are deeply related: $s+u-h=2$,
thus the pair $\{s,u\}$ defines unambiguously the \emph{primary equilibrium class} of the body \cite{Domokos2006}.
The primary equilibrium classes can be refined into \emph{secondary equilibrium classes} which determine also the topology of the equilibria defined by the Morse--Smale complex of the height function $R$ \cite{Domokos2012}.
This topology is usually given by a 3-coloured spherical quadrangulation,
where the vertices are the equilibria, the colour of a vertex defines its equilibrium type (stable, unstable or saddle), and
an edge is a particular path on the surface connecting two equilibria, whose tangent vectors agree with the gradient vectors of $R$ (\autoref{fig:horbit}).
We transform these 3-coloured surface graphs with a bijection to generic 2-coloured quadrangulations keeping the underlying geometric meaning, and we use this latter form to define the secondary class as a fully combinatorial object.
We mention that a spherical quadrangulation is always bipartite \cite{Archdeacon2001,Gross2004} hence the 2-colouring is possible.

The data set yielded by our program enables us to present some statistics on number of possible secondary classes as well.
Because of the 1-1 correspondence of 2-coloured quadrangulations 
 and connected plane graphs \cite{Fusy2007,Brinkmann2005}, 
our numbers agree with the numbers of the unrooted and unsensed maps calculated by \citet{Wormald1981}, and exhaustively enumerated by \citet{Walsh1983, Walsh2012}.
\citet{Walsh2007} surveys different types of map census results achieved both exhaustive search and with formulae.
We are also able to give statistics on the number of multiquadrangulations using a trivial relation between the number of maps, quadrangulations  and self-dual quadrangulations for a fixed size.

\citet{Domokos2006} constructed the geometry of a representative body of the primary class $\{1,1\}$ referred as mono-monostatic body, also known as \emph{Gömböc}.
They also concluded that every primary class for all $s,u\geq 1$ is generated from the \Gomboc via their specific geometric transformations called \emph{Columbus' algorithm}. 
Columbus' algorithm is a sequence of sensitive modifications on a body perturbing the surface only at the vicinity of an equilibrium, such that the body belonging to the primary class $\{s,u\}$ is transformed to another one belonging to $\{s+1,u\}$ or $\{s,u+1\}$ (illustrated on \autoref{fig:colgeo}).
\begin{figure}
  \centering
  \subfigure[Cutting the surface with a plane near $U_0$]{
    \label{fig:colgeo1}
    \qquad
    \includegraphics[width=5cm,bb=11 02 1199 843]{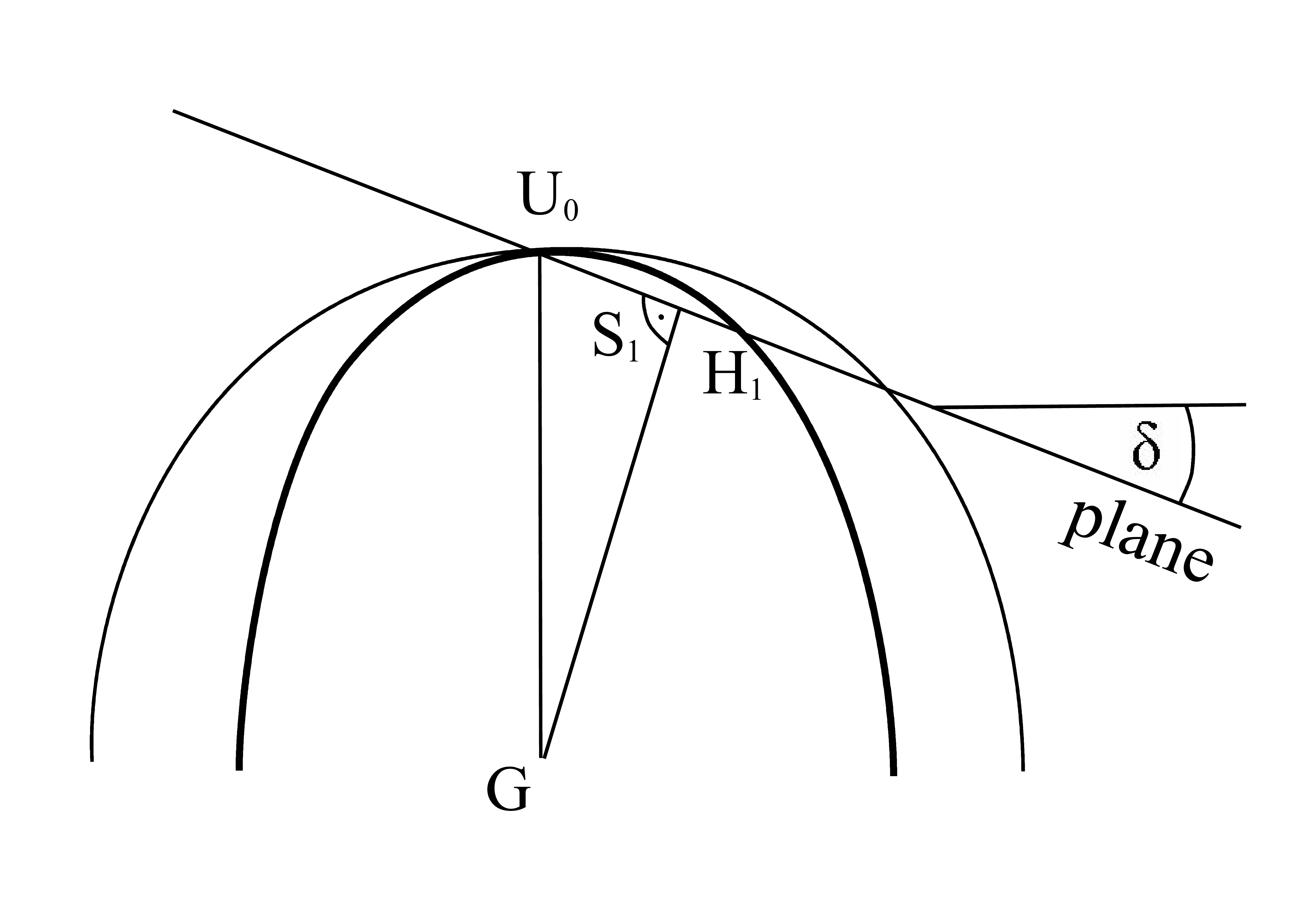}
  } 
  \qquad
  \subfigure[Modified gradient field of $-R$]{
    \quad
    \label{fig:colgeo2}
    \includegraphics{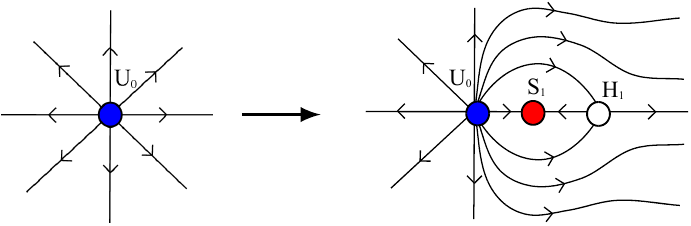}
    \quad
  }
  \caption{Applying a step of Columbus' algorithm on an egg-shaped surface.
Two equilibria appear on the surface: $S_1, H_1$.}
  \label{fig:colgeo}
\end{figure}
Because of the geometric feasibility of Columbus' algorithm, i.e.~the transformations always can be applied around any equilibrium, they referred to the \Gomboc as the ancestor of every primary class.

The combinatorial equivalent of the original Columbus' algorithm is the sequence of monotone coloured splittings.
The result of \cite{Domokos2006} can be reformulated as follows: for any $s,u$, some secondary class in the primary class $\{s,u\}$ is generated from the \Gomboc using the monotone coloured splittings.
(We detail the secondary class of \Gomboc in \autoref{sec:geometry}.)
We extend this statement with the following corollaries of our theorems:
\begin{coro}
  \label{coro:irreduc}
  Not every secondary class is generated from the \Gomboc using the monotone coloured splittings, but the ones with $s+u<8$ are all generated from the \Gomboc.
\end{coro}
  We say a polyhedron is a \emph{minimal polyhedron} if its every face contains one stable and its every vertex is an unstable equilibrium.
With respect to the monotone coloured splittings (and the original Columbus' algorithm), they are also ancestors:
\begin{coro}
  \label{coro:minpol}
  The secondary classes of the minimal polyhedra are irreducible so not generated from any other secondary class.
\end{coro}
Finally, if we aim to generate a secondary class from every primary class, it is enough to use one of the two monotone coloured splittings:
\begin{coro}
  \label{coro:restr}
  It is enough to use either the coloured splitting $S_{1,1}$ or $S_{2,2}$ to generate every primary class from a finite starting set of secondary classes.
\end{coro}

We mention that the coloured splittings are also a restricted subset of the operations called ``cancellations'' in computational geometry by \citet{Edelsbrunner2001} and \citet[Fig.~7.3]{Bremer2004}, who used them to simplify a multi-resolution mesh structure.
Domokos, Lángi and Szabó \cite{Domokos2012} showed recently that the geometric counterparts of the unrestricted coloured splittings generate the whole family of secondary classes, however, this is beyond the scope of this paper.

\subsection{Organization of this paper}

The paper is organized as follows.
Section~\ref*{sec:splitting} analyzes the vertex splitting and establishes its properties on the reachability of quadrangulations.
Section~\ref*{sec:irreducible} presents the hierarchy of irreducible quadrangulations: the reachability of small graphs, the polyhedral irreducible quadrangulations and the uniqueness of ancestors.
Section~\ref*{sec:geometry} interprets our results on the equilibrium topologies of convex bodies.
Finally we show some statistics on our data set in \autoref*{sec:stats}.

\section{The properties of vertex splitting} 
\label{sec:splitting}


In order to generalize the splitting for parallel case (with the notations illustrated on \autoref{fig:splitting}), 
let $\sigma(v)$ denote the cyclic ordering of the vertex $v$ in clock-wise order, so $\sigma(v)=\left(e_1,\ldots,e_{d(v)}\right)$. 
Let $n_i$ denote the other endpoint of the edge $e_i$.
\begin{defn}
A \emph{vertex splitting} on a quadrangulation $G$ is a graph operation transforming $G$ to $G'$, specified by a walk $n_1e_1ve_mn_m$ of $G$.
The vertex $v$ of $G$ is replaced by vertices $w$ and $v'$ in $G'$ dividing the edges of $v$.
The cyclic ordering of $w$ is $\sigma(w)=(e'_1,\ldots,e'_\sdeg)$, where $e'_i$ has the same other endpoint as $e_i$ had;
the cyclic ordering of $v'$ is $\sigma(v')=(e_1,e''_\sdeg,e_{\sdeg+1},\ldots,e_{d(v)})$, keeping some edges of the former vertex $v$, and $e''_{\sdeg}=e_\sdeg$ if $\sdeg>1$, otherwise $e''_\sdeg$ has the same other endpoint as $e'_\sdeg$ and $e_\sdeg$ had.
\end{defn}

Note that if $m>1$ and $n_1$ is identical to $n_\sdeg$ then $e_1$ and $e_\sdeg$ are distinct parallel edges.
The new face introduced by the splitting is bounded by the walk $v'e_1n_1e'_1we'_mn_me''_mv'$.
Observe that the degree of the new vertices are $d(w)=\sdeg$ and $d(v')=d(v)-\sdeg+2$.
In addition, the splitting specified by the walk $n_1e_1ve_\sdeg n_\sdeg$ creates the very same graph that the one specified by the walk $n_\sdeg e_\sdeg ve_1 n_1$.
Hence the \emph{degree of a splitting}, denoted by $D$, is invariant to reflection, and $D:=\min\{d(v'),d(w)\}=\min\{m,d(v)-m+2\}$.

Let $n$ denote the number of vertices of the quadrangulation.
There are exactly four quadrangulations such that $n\leq 4$, shown on \autoref{fig:smallests}, where $C_4$ is generated from $P_2$ with a 2-splitting, and $Q_3,Q_4$ are generated from $P_2$ with a 1-splitting.
For larger graphs, we have

\begin{prop}
\label{lem:mindeg}
For a quadrangulation such that $n>4$,
\begin{enumerate}[(i)]
\item 
  \label{lem:mindeg:deg1}
  if it has a vertex of degree 1, then it is 1-contractible even for $n=4$,
\item
  \label{lem:mindeg:deg2}
  if it has a vertex of degree 2, then it is 2-contractible,
\item
\label{lem:mindeg:degk}
  if its minimum degree is $k$, then it is $k$-contractible, and not $l$-contractible for any $l<k$.
\item
\label{lem:mindeg:irred}
  it is irreducible if and only if its minimum degree is 3,
\end{enumerate}
\end{prop}
\begin{figure}
\centering
\subfigure[$P_2$]{
  \includegraphics{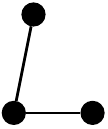}
}
\qquad
\subfigure[$C_4$]{
  \includegraphics{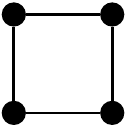}
}
\qquad
\subfigure[$Q_3$]{
  \includegraphics{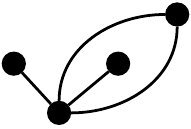}
}
\qquad
\subfigure[$Q_4$]{
  \includegraphics{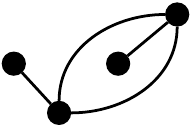}
}
\caption{The four smallest quadrangulations with $n\leq 4$.}
\label{fig:smallests}
\end{figure}
\begin{proof}
 Part (\ref*{lem:mindeg:deg1}) is a straightforward consequence of the quadrangulated faces.

For part (\ref*{lem:mindeg:deg2}), suppose a vertex $w$ is incident to two edges going to $n_1$ and $n_2$.
If $n_1$ and $n_2$ are distinct vertices, then the graph is 2-contractible as depicted on \autoref{fig:monotone}.
If $n_1=n_2$ and $n_1$ was not incident to any other parallel edge, the graph would be isomorphic to $Q_3$ (see \autoref{fig:smallests}).
If there are other parallel edges incident to $n_1$, then the graph is 2-contractible, the parallel case is applicable (see \autoref{fig:monotone}).

\begin{figure}
  \centering
  \includegraphics{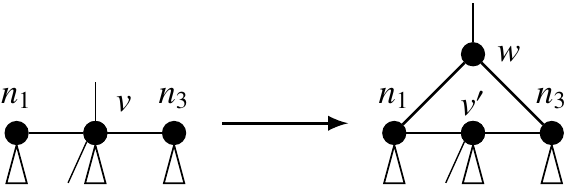}
\qquad
  \includegraphics{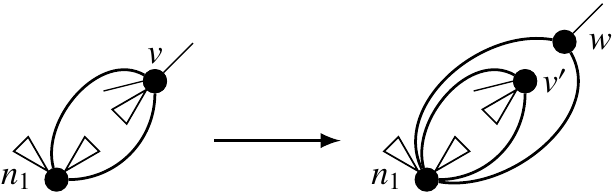}
  \caption{3-splitting. The half edge at $v$ denotes that there must be an edge otherwise it would be a 2-splitting.}
  \label{fig:3splitting}
\end{figure}

Part (\ref*{lem:mindeg:degk}) follows from the definition of the degree of splitting and from the observation that a $k$-splitting introduces a vertex of degree $k$.

Part (\ref*{lem:mindeg:irred}) follows from (\ref*{lem:mindeg:deg1})-(\ref*{lem:mindeg:degk}) as the monotone splittings introduce a vertex of degree 1 or 2.
\end{proof}

\autoref{lem:mindeg} provides the necessary observations to prove that the splittings $S_{1,3}$ generate $\mathscr{Q}$.

\begin{proof}[Proof of \autoref*{theo:splitting}]
The statements follows from \autoref{lem:mindeg} and from the fact that the minimum degree of a multiquadrangulation is either 1, 2 or 3 by Euler's formula.
\end{proof}

It is clear that the 1-splitting and the 2-splitting are monotone, and modify only one face of the graph in the vicinity of a point.
Hence the perturbation of the graph is minimal.
Monotonicity is also sufficient to characterize these splittings:

\begin{prop}
  \label{prop:monotone}
  Every plane graph operation defined on quadrangulations which adds one vertex and is monotone, is a restricted splitting $S_{1,2}$.
\end{prop}

\begin{proof}
  Consider any operation satisfying the conditions, putting the new vertex $w$ into an original quadrilateral face bounded by the walk $x_1x_2x_3x_4x_1$, listing only the vertices of the walk.
  Note that some of the vertices $x_i$ may coincide.
  By Euler's formula, we need to add two edges in addition to the new vertex, so connect, without loss of generality, $w$ and $x_1$ with an edge.
  Then, we have an ``almost quadrangulation'': except for one face which is bounded by the walk $x_1wx_1x_2x_3x_4x_1$ of length 6.
  This face can be divided into two quadrilateral faces by adding one edge in three ways
: connecting $x_1$ and $x_2$ (the fourth and the seventh elements of the walk sequence), or $x_1$ and $x_4$ (third and sixth elements), or $w$ and $x_3$ (second and fifth) with an edge.
  The first and the second case is applying a 1-splitting, the third case is applying a 2-splitting.
\end{proof}

\section{Irreducible quadrangulations}
\label{sec:irreducible}

As being irreducible for $n>3$ is equivalent to having minimum degree of 3, one can easily verify that e.g.\  the radial graph of the skeleton of the tetrahedron is irreducible (see \autoref{fig:tetrahedron}).
\begin{figure}
\centering
  \subfigure[Radial graph (left) and the skeleton (right) of the tetrahedron]{
    \label{fig:tetrahedron}
    \includegraphics{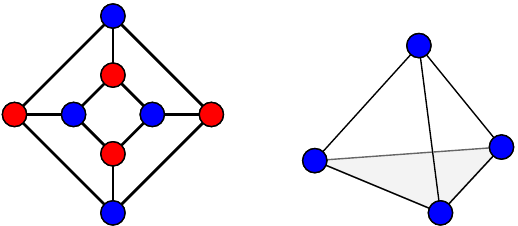}
  }
\qquad\qquad\qquad
  \subfigure[Radial graph (left) and the skeleton (right) of the square pyramid]{
    \label{fig:sqpyramid}
    \includegraphics{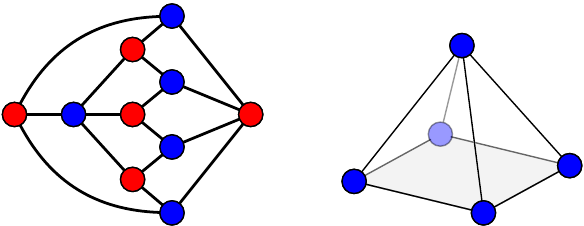}
  }
\caption{Irreducible quadrangulations.}
\label{fig:irreducible}
\end{figure}
According to \autoref{theo:gomboc}, it is the smallest irreducible quadrangulation after $P_2$.
In preparation of proving it, we need to generalize slightly an earlier result:
\begin{prop}[\citet{Batagelj1989}]
  \label{prop:Batagelj}
  Every simple 3-connected quadrangulation has at least 8 vertices of degree 3.
\end{prop}

The generalization of Batagelj's proof for our need is straightforward as his statement still holds for multigraphs and instead of 3-connectivity it is enough to assume the minimum degree is 3.
The latter observation has already been made in \cite{Brinkmann2005}.
Hence we use the following statement:

\begin{prop}
  \label{prop:Batageljgen}
  Every quadrangulation with minimum degree 3 has at least 8 vertices of degree 3.
\end{prop}

\begin{proof}
The layout of this proof is similar to Batagelj's proof, but the conditions are relaxed.
From Euler's formula we have that the sum of the degrees of a quadrangulation of size $n:=s+u$ is $4n-8$.
If it has $k$ vertices of degree 3 and $n-k$ other vertices of degree at least 4, then the sum of the degrees is also at least $3k+4(n-k)$, implying $k\geq 8$.
\end{proof}

\begin{proof}[Proof of \autoref*{theo:gomboc}]
  By \autoref{lem:mindeg} and \autoref{prop:Batageljgen}, we have that the only irreducible ancestor such that $n<8$ is $P_2$.
So the theorem can be proved by induction.
\end{proof}

The data set generated by our program has strengthened \autoref{theo:gomboc}, as we could count easily the irreducible ones for $n\leq 10$ by \autoref{lem:mindeg}.
The statistics showed that there are exactly three irreducible graphs among them: $P_2$ and the radial graphs of the two polyhedral skeletons shown on \autoref{fig:irreducible}.
\autoref{theo:minpol} states that the radial graphs of polyhedral skeletons are all irreducible.
Moreover, it gives an inductive definition for them, as they cannot be generated by monotone splittings.

\begin{proof}[Proof of \autoref*{theo:minpol}]
Any polyhedral skeleton is simple and 3-connected by Steinitz' theorem \cite{Mohar2001}.
Let $\mathscr{Q}_4$ denote the family of simple and 3-connected quadrangulations having no separating 4-cycles.
It is known \cite{Brinkmann2005} that a graph is simple and 3-connected if and only if its radial graph belongs to $\mathscr{Q}_4$, proving the first part of the theorem.

It was also shown \cite{Brinkmann2005} that $\mathscr{Q}_4$ is generated by the restricted splitting $S_{3,3}$ starting from the pseudo-double wheels.
A \emph{pseudo-double wheel} is a cycle of even length, with its inner and outer face subdivided by a vertex, such that the inner vertex is adjacent to the odd-numbered vertices of the cycle, the outer vertex is adjacent to the even-numbered vertices of the cycle \cite{Brinkmann2007}.
It is easy to verify that a pseudo-double wheel of size $n=2k+2$ is the radial graph of the skeleton of a pyramid with a $k$-sided base which completes the proof.
\end{proof}

Now we show that the irreducible ancestor of a quadrangulation is unambiguously defined, so a quadrangulation cannot be generated from different irreducible ancestors at the same time.
Consequently, for any two different irreducible quadrangulations $A_1\neq A_2$ the families generated starting from $\{A_1\}$ and from $\{A_2\}$ are disjoint, thus the set of irreducible quadrangulations leads to a natural partition of $\mathscr{Q}$.

\begin{proof}[Proof of \autoref*{theo:ancestor}]
Without loss of generality, suppose $A_1$ is not isomorphic to $P_2$ thus the minimum degree of $A_1$ is 3.
Let $G_0,G_1,\ldots,G_p$ be the series of graphs yielded by the generation process of $G$ from $A_1$ where $G_0= A_1$ and $G_p= G$, and, indexed in reverse order, $G_{p+q},\ldots,G_{p+1},G_p$ the graphs yielded by the generation process from $A_2$ where $G_{p+q}= A_2$, for some $p,q>0$.
We will prove that the edges of $G_0$ remain intact despite applying any monotone vertex splitting or monotone face contraction,
i.e.~$G_0$ is the embedded subgraph of every $G_k$ for all $0\leq k \leq p+q$, so $A_1$ is the subgraph of $A_2$, and vice versa.

We say an edge of a graph $G_i$ is \emph{ancient} if it is also in $G_0$.
By induction, assume that $G_0$ is the embedded subgraph of the graphs $G_0,\ldots,G_{k-1}$ for some $k$, which trivially holds for $k=1$.
If $G_k$ is created from $G_{k-1}$ by a monotone splitting, 
clearly no edge is removed due to the definition of monotonicity.
If $G_k$ is created from $G_{k-1}$ by a 2-contraction, then we use a basic property of the 2-contraction that if it removes an edge, then one endpoint of the edge is also removed and the degree of the removed endpoint is 2.
However, according to the induction hypothesis, each ancient edge of $G_{k-1}$ has endpoints with degree at least 3 thus no 2-contraction is applicable to $G_{k-1}$ to remove an ancient edge.

The last case is where $G_k$ is created from $G_{k-1}$ by a 1-contraction, removing two edges: $e'_1$ and $e''_1$ (see the definition of vertex splitting and \autoref{fig:1splitting}).
Edge $e'_1$ cannot be ancient because of having and endpoint of degree 1 as explained in the previous case.
If edge $e''_1$ is ancient and $e_1$ is not, then again no ancient edge is removed as it is only a technicality if we actually remove the edge $e''_1$ or $e_1$ because of the symmetry of the 1-contraction.
If both edges $e''_1$ and $e_1$ would be ancient and $e'_1$ would not be, then, because of the induction hypothesis, there would be a face in $G_0$ bounded by the walk $v'e_1n_1e''_1v'$ contradicting the assumption that $G_0$ is a quadrangulation.
Thus a 1-contraction cannot remove an ancient edge either.

Consequently, as neither the splittings nor the contractions alter the embedding of the rest of the graph, $G_0$ remains the embedded subgraph of $G_k$ for all $k$.
\end{proof}

\section{Generating secondary equilibrium classes}
\label{sec:geometry}

In this section we interpret our results on generating secondary equilibrium classes as coloured surface graphs of convex bodies.
Some of these observations were already outlined in the conference version of this paper \cite{Kapolnai2011}.

The idea of representing a 3D body with some surface graph appears in multiple disciplines.
Most famous are the polyhedral graphs that are the skeletons of the convex polyhedra, characterized by Steinitz's theorem~\cite{Mohar2001}.
Another related appearance is the mesh generation of physical shapes in computational geometry, when some Morse--Smale complex of a body is drawn as a surface graph, see e.g.\  \cite{Dong2006}.
However, the underlying function used for meshing is not necessarily the height function $R$, thus the nodes of the mesh and the equilibrium points do not necessarily coincide.

\subsection{Geometric interpretation of quadrangulations and vertex splittings}

To determine the secondary class of a body, we need to introduce some concepts from Morse theory~\cite{Arnold1978,Edelsbrunner2001}, illustrated on~\autoref{fig:horbit}.
In generic case, we say a path on the surface is a \emph{heteroclinic orbit}, if its tangent vectors agree with the gradient vectors of the height function $R$, and its endpoints are two equilibria of different type (see \autoref{fig:ellorbits} and \ref{fig:ellorb}).
It is known the heteroclinic orbits incident to a saddle point are isolated on the surface, and there are only a finite number of them.

These isolated orbits divide the body surface into quadrilateral cells (see \autoref{fig:msgraph}), and in each cell an infinite number of non-isolated, heteroclinic orbits are going from the unstable to the stable point which we disregard for now.
In this way, the body surface defines a vertex-coloured multigraph embedded on the sphere, where the vertices are the equilibria, the edges are the isolated heteroclinic orbits connecting saddle and non-saddle vertices, the faces are the quadrilateral cells, and the colour of a vertex gives its type of equilibrium.

This graph is referred as the \emph{Morse--Smale graph}, which means possessing three properties~\cite{Arnold1978,Edelsbrunner2001}.
(1) The graph is a \emph{quadrangulation} of the plane unless the body is in the primary class of the \Gomboc.
(2) Every quadrilateral boundary walk is a sequence of a saddle, a stable, a saddle and an unstable vertex.
(3) The degree of every saddle vertex is 4.

\begin{figure}
  \centering
  \subfigure[Morse--Smale graph (isolated orbits)]{
    \begin{minipage}{0.27\linewidth}\centering
      \label{fig:msgraph}
      \includegraphics{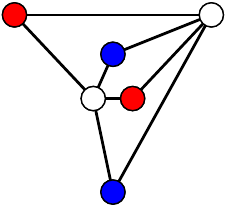}\vspace{0.5em}
    \end{minipage}
  }
  \subfigure[Triangulation]{
    \begin{minipage}{0.27\linewidth}\centering
      \label{fig:trigraph}
      \includegraphics{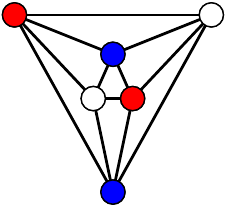}\vspace{0.5em}
    \end{minipage}
  }
  \subfigure[Quasi-dual]{
    \begin{minipage}{0.27\linewidth}\centering
      \label{fig:topgraph}
      \includegraphics{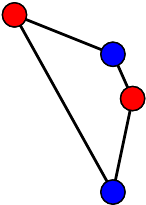}\vspace{0.5em}
    \end{minipage}
  }
  \caption{Determining the secondary class of the ellipsoid.
  }
  \label{fig:subclass}
\end{figure}

We perform two invertible transformations on the Morse--Smale graph, to yield a more compact graph which is still a genuine representation of the geometry.
First, we connect the stable and unstable points in each quadrangulated face, creating a triangulation (see \autoref{fig:trigraph}).
The geometric interpretation of this operation could be that from each cell, we pick one orbit arbitrarily from the infinite set of unstable-stable orbits.
Then we remove the saddle points and the edges incident to them, creating another quadrangulation leaving only the edges created in the previous step (see \autoref{fig:topgraph}).
This idea also appears in \cite{Dong2006} as the \emph{quasi-dual} of the Morse--Smale complex, however, as their goal is to simplify of the mesh structure, they do not mention that it can be shown that these transformations are invertible, so we lost no information by removing the saddles.
We use the resulting quasi-dual graph to define the secondary class of the body, summarized in the following

\begin{defn}
  A \emph{quasi-dual} is a 2-coloured multiquadrangulation, where the independent sets are called the stable and the unstable equilibria, of size $s$ and $u$, respectively, if $s+u>2$.
  A \emph{quasi-dual of a body} is the quasi-dual obtained with the process described above.
  For technicality, in case $s+u=2$, we define the quasi-dual of the \Gomboc as the path $P_1$ of length 1 (with one edge) connecting a stable and an unstable equilibrium (see \autoref{fig:gomboc}).
  Obviously $n=s+u$.
  A \emph{secondary equilibrium class} is an isomorphism class of quasi-duals, where the isomorphism is expected to preserve the colouring as well.
\end{defn}

\begin{figure}
  \centering
  \includegraphics{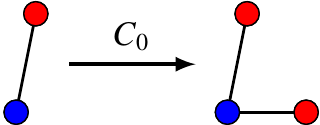}
  \qquad
  \includegraphics{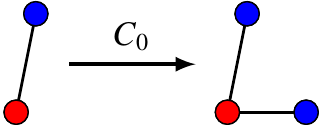}
  \caption{Columbus' algorithm: auxiliary coloured splitting $C_0$ defined only on the quasi-dual $P_1$ of the \Gomboc.}
  \label{fig:gomboc}
\end{figure}

While the derivation process above requires the existence of the gradient field of the height function, we would also like to consider some reasonable body surfaces with no gradient, e.g.\  a polyhedra.
As the construction of the Morse--Smale complex is already extended for some non-smooth functions as well~\cite{Edelsbrunner2001}, we believe the definition of the quasi-dual and the secondary class of a polyhedron could be understood by intuition and needs no rigorous theoretical background.

If a quadrangulation admits multiple non-isomorphic colourings, then its possible secondary classes form a partition of its map. 
It is easy to verify that a quadrangulation admits either two non-isomorphic colourings, i.e.\  switching the colours results in another secondary class, or only one colouring.
In the latter case we call the secondary class a \emph{self-dual}.

Assigning a quasi-dual to a convex body is a generalization of assigning a radial graph to a polyhedral skeleton.
Radial graphs are all quadrangulations, and every quadrangulation is the radial graph of some spherical surface graph~\cite{Mohar2001, Brinkmann2005}.
Moreover, assigning to a generic connected plane graph $G$ its 2-coloured radial graph $R(G)$ and colouring the images of the vertices of $G$ ``unstable'' in $R(G)$, and the images of the faces of $G$ ``stable'' defines a bijection \cite{Fusy2007}.
For a fixed $s,u$, the bijection is between the set of secondary classes in the primary class $\{s,u\}$ and the set of maps of generic connected plane graphs with $u$ vertices, $s$ faces and $h=s+u-2$ edges.
E.g.\  in the case of the \emph{minimal polyhedra}, the quasi-dual is actually the coloured radial graph of the skeleton, thus the Poincaré--Hopf theorem can be replaced by Euler's formula.

The \emph{steps of Columbus' algorithm} are defined and their geometric feasibility were proved in \cite{Domokos2006}, where feasibility means they transform any convex body to another convex one.
Observing the change of the quasi-dual of the underlying body, the steps of Columbus' algorithm are combinatorial operators on quasi-duals we call the \emph{coloured splittings}.
A coloured splitting of a quasi-dual consists of a vertex splitting of the underlying quadrangulation followed by the proper colouring of the introduced vertex.
We add that the coloured splittings corresponding to the original steps of Columbus' algorithm are all monotone.
So each coloured splitting corresponds to two dual versions of a splitting of the quadrangulation: to one adding a new stable vertex, and to one adding a new unstable vertex.

For technicality, we add an auxiliary coloured splitting $C_0$ applicable only on the quasi-dual $P_1$ of the \Gomboc, which is not literally a vertex splitting of a quadrangulation.
Applying $C_0$ either results in the quasi-dual in the class $\{1,2\}$, or in the quasi-dual in the class $\{2,1\}$, see \autoref{fig:gomboc}.
So the definition of coloured splitting is relaxed to allow $C_0$ as well, considered a 1-splitting (i.e.\  $S_{1,1}$).

\subsection{Generating secondary classes by monotone coloured splittings}

While the monotone coloured splittings generate every primary class from $P_1$, they do not generate every secondary class, so they admit a nontrivial family of irreducible ancestors.
We say a secondary class is \emph{irreducible} if it cannot be created with a monotone coloured splitting.
In other words its underlying quadrangulation with $n>3$ is irreducible, or $n=2$.
So the secondary classes in $\{1,2\}$, $\{2,1\}$ are not considered irreducible (see \autoref{fig:gomboc}), but $P_1$ in $\{1,1\}$ is.
By \autoref{lem:mindeg}, we can characterise the irreducible secondary classes for $n>3$ as bodies with a quasi-dual of minimum degree 3.

The results of this subsection are outlined in \autoref{fig:suh-outline} which we reveal gradually. 
Let $\mathcal{G}$ denote the starting set consisting of the secondary class of the \Gomboc, i.e.\  $\mathcal{G}=\{P_1\}$, and \col the family generated from $P_1$ by monotone coloured splittings.
Throughout this section, for any starting set of secondary classes $X$, let $X^+$ denote the family of secondary classes generated from $X$ by the monotone coloured splittings.
So let $\mathcal{I}$ denote the family of irreducible secondary classes, then \suh denotes the family of all possible secondary classes.

\begin{figure}
  \centering
  \includegraphics{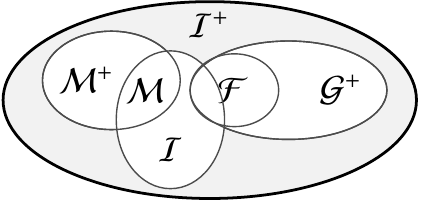}
  \caption{Hierarchy of secondary classes.
    Notation: \suh: all secondary classes,
    \col: generated from $P_1$,
    $\mathcal{F}$: small secondary classes s.t.~$s+u<8$,
    $\mathcal{I}$: irreducible ancestors,
    $\mathcal{M}$: minimal polyhedra,
    $\mathcal{M}^+$: generated from minimal polyhedra.}
  \label{fig:suh-outline}
\end{figure}

The first part of \autoref{coro:irreduc} states that $\mathcal{I}$ contains secondary classes besides $P_1$, which clearly follows from the fact that there is an infinite set of non-isomorphic irreducible quadrangulations.
For the second part, let $\mathcal{F}$ denote the family of secondary classes in the primary classes $\{s,u\}$ such that $s+u<8$.
The second part states that $\mathcal{F}\subset\colm$, which follows from \autoref{theo:gomboc} with the help of the auxiliary splitting $C_0$. 

Note that $\colm\cap\mathcal{I}$ contains only $P_1$ by definition.
\autoref{coro:minpol} gives a geometric characterisation of a special family within $\mathcal{I}$ called the minimal polyhedra.
Recall that a polyhedron is a \emph{minimal polyhedron} if its every face contains one stable and its every vertex is an unstable equilibrium, and let $\mathcal{M}$ denote the family of their secondary classes.
E.g.\  the Platonic solids or a right prism are all minimal polyhedra, however, an oblique prism may not be a minimal polyhedron: there may be a face which does
not contain a stable point.
\autoref{coro:minpol} states that $\mathcal{M}\subset\mathcal{I}$  (see \autoref{fig:suh-outline}).
As there are only three irreducible secondary classes until $s+u\leq 10$ according to our data set, we can safely say that every secondary class such that $s+u\leq 10$ are generated either from the \Gomboc or from a minimal polyhedron, i.e.\  $\{S\in\suhm:s+u\leq 10\}\subset\colm\cup\mathcal{M}^+$.

Nevertheless, there are other irreducible secondary classes as well, e.g.\  \autoref{fig:parallelirred} shows a method to enumerate some of them with parallel edges.
Clearly a minimal polyhedron cannot have parallel edges.

\begin{figure}
\centering
\includegraphics{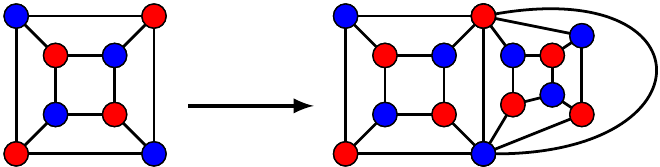}
\caption{Irreducible subclass with parallel edges created by reflecting a simple one.}
\label{fig:parallelirred}
\end{figure}

Finally, \autoref{theo:ancestor} states that the irreducible ancestor of a secondary class is unique, because the quadrangulations of the irreducible secondary classes and the irreducible quadrangulations coincide for $n>3$.
Consequently, e.g.\  \col and $\mathcal{M}^+$ are disjoint.
\autoref{theo:ancestor} also results in an algorithm to determine the irreducible ancestor of a body:
we simply need to apply monotone contractions iteratively to its quasi-dual, until we get the irreducible ancestor.
This algorithm is trivial because we can choose the contractions arbitrarily: the resulting irreducible ancestor is unique hence is independent of the actual choices, as any contraction sequence would result in the same ancestor.

\subsection{Generating primary classes by restricted splittings}

The restricted coloured splittings $S_{1,2}$ (i.e.\  the monotone ones) can generate every primary class from the \Gomboc \cite{Domokos2006}, but not every secondary class by \autoref{coro:irreduc}.
We consider further restrictions in this subsection summarized by \autoref{coro:restr}.
In detail, the coloured splitting $S_{1,1}$ generates every primary class from $P_1$, and the coloured splitting $S_{2,2}$ generates every primary class from the starting set of the secondary classes of the singleton primary classes $\left\{\{1,1\},\{2,1\}, \{3,1\},\{1,2\},\{1,3\}\right\}$.
Note that according to \autoref{tab:cardinalities} in \autoref{sec:stats} these are the only singleton primary classes.

\begin{proof}[Proof of \autoref*{coro:restr}]
First we consider restricting the monotone coloured splittings to $S_{1,1}$.
As 1-splitting is always applicable at any part of the graph if $n>2$, any primary class can be reached with coloured 1-splittings from $\{1,2\}$ or from $\{2,1\}$.
So the primary classes are generated from the starting set $\{\{1,1\},\{1,2\},\{2,1\}\}$, and this set can be generated from $P_1$ by $C_0$.

Now we consider restricting the monotone coloured splittings to $S_{2,2}$.
The quadrangulations of the classes $\{2,1\}$, $\{3,1\}$, $\{1,2\}$, $\{1,3\}$ are 2-irreducible (see \autoref{fig:smallests}), so they has to be in the starting set.
The 2-splitting is always applicable around a vertex $v$ with $d(v)\geq 2$ introducing a vertex $w$ of the same colour as $v$.
It is clear that if $n>3$, both independent sets of a bipartite quadrangulation contains vertices of degree at least 2.
As $C_4$ in $\{2,2\}$ is 2-contractible, the 2-splittings can be arbitrarily combined to generate some secondary class in any primary class $\{s,u\}$ such that $s+u>4$.
\end{proof}

In the next section we show statistics on the number of 1-irreducible and 2-irreducible secondary classes for a limited size in \autoref{tab:restricted}.

\section{Computational results}
\label{sec:stats}

This section presents some statistics on the data set attained by the computer program.
As \plantri supports dividing the computation into independent parts, we could perform this computation in parallel in a grid infrastructure using the Saleve framework~\cite{Dobe2010}.
Some of the numbers have been already presented in \cite{Kapolnai2011}.

\autoref{tab:cardinalities} shows the cardinalities of the classes $\{s,u\}$ with $s+u\leq 10$.
The table is symmetric because a subclass such that $s\neq u$ is clearly not self-dual.
These numbers were already published as the cardinalities of the 
 \emph{unrooted and unsensed maps} by \citet{Wormald1981, Walsh1983, Walsh2012}.

\autoref{tab:q} shows the number of multiquadrangulations $\quads(n)$, which is the first publication of these numbers, up to our best knowledge.
To compute $\quads(n)$ from the same program-generated data set, 
observe a relation between the number of secondary classes in the class $\{s,u\}$, the number of self-dual secondary classes and the number of quadrangulations of size $n=s+u$ and, denoted respectively by $\maps(s,u)$, $\sdquads(s,u)$, $\quads(n)$:
\begin{equation}
\label{eq:card}
2\quads(n)-\sdquads\left(n/2,n/2\right)=\sum_{s=1}^{n-1}\maps(s,n-s).
\end{equation}

\autoref{tab:restricted} shows the number of ancestor secondary classes with respect to splittings with different restriction criteria: $S_{2,2}$, $S_{1,1}$ and $S_{1,2}$.

\begin{table}[p]
  \centering
  \begin{tabular}{rrrrrrrrrr}
    \hline
    {\ }  & $s=1$
    & $s=2$ & $s=3$ & $s=4$ & $s=5$ & $s=6$ & $s=7$ & $s=8$ & $s=9$ \\
    \hline
    $u=1$& {1} & \nfour 1   & 1    & \neight 2    & 3   & \ntwelve 6    & 12  & \nsteen 27 & 65\\
    $u=2$&\nfour 1  & 2   & \neight 5    & 13   & \ntwelve 35  & 104  & \nsteen 315 &1021\\
    $u=3$& 1  & \neight 5   & 20 & \ntwelve 83  & 340 & \nsteen 1401 &5809\\
    $u=4$&\neight 2  & 13  & \ntwelve 83   & 504 & \nsteen 2843 &15578\\
    $u=5$& 3  & \ntwelve 35  & 340  & \nsteen 2843 &21420\\
    $u=6$&\ntwelve 6  & 104 & \nsteen 1401 &15578\\
    $u=7$& 12 & \nsteen 315 &5809\\
    $u=8$&\nsteen 27 &1021\\
    $u=9$ &65\\
    \hline
  \end{tabular}
  \caption{Cardinalities $\maps(s,u)$ of the equilibrium classes.}
  \label{tab:cardinalities}
\end{table}

\begin{table}
  \centering
  \begin{tabular}{lrrrrr}
    \hline
    & $\quads(n)$ & $\sdquads(n)$ & $\sum_s\maps(s,n-s)$\\
    \hline
    $n=3$  &     1 &     - &     2\\
    $n=4$  &     3 &     2 &     4\\
    $n=5$  &     7 &     - &    14\\
    $n=6$  &    30 &     8 &    52\\
    $n=7$  &   124 &     - &   248\\
    $n=8$  &   733 &    50 &  1416\\
    $n=9$  &  4586 &     - &  9172\\
    $n=10$ & 33373 &   380 & 66366\\
    \hline
  \end{tabular}
  \caption{The number of multiquadrangulations ($\quads$), self-dual secondary classes ($\sdquads$) and secondary classes ($\sum e$).}
  \label{tab:q}
\end{table}

\begin{table}
  \centering
    \begin{tabular}{lrrrr}
      \hline
       & 1-2-contractible & 2-irreducible & 1-irreducible & irreducible\\
      \hline
      $n=4$  &     0  &   3 &     1 & 0 \\
      $n=5$  &     6  &   2 &     6 & 0 \\
      $n=6$  &    32  &   4 &    16 & 0 \\
      $n=7$  &   172  &  10 &    66 & 0 \\
      $n=8$  &  1071  & 33  &   311 & {1}\\
      $n=9$  &  7370  & 114 &  1688 & 0 \\
      $n=10$ &  55766 & 474 & 10125 & {1}\\
      \hline
    \end{tabular}
  \caption{Ancestor secondary classes with respect to $S_{1,1}$ (third column), $S_{2,2}$ (second column), $S_{1,2}$ (fourth column) and non-ancestors (first column).}
  \label{tab:restricted}
\end{table}






\section{Acknowledgements}

Research is supported by the Hungarian NTP program TECH\_08-A2/2-2008-0097 (WEB2GRID), and by OTKA grant 104601.



\clearpage

\bibliographystyle{abbrvnat}
\bibliography{mybib}













\end{document}